\documentclass{elsarticle}
\usepackage[T1]{fontenc}
\usepackage{anysize}
\marginsize{2cm}{2cm}{2cm}{2cm}
\usepackage[utf8]{inputenc}
\usepackage{latexsym,amssymb,amsmath}
\usepackage{color}
\usepackage{scalerel,hyperref,url}
\usepackage{verbatim}
\usepackage{algorithm}
\usepackage{algorithmic}
\usepackage{makecell}
\usepackage{multirow}

\usepackage{amsthm}
\newtheorem{definition}{Definition}
\newtheorem{proposition}{Proposition}
\newtheorem{lemma}{Lemma}

\newcommand{\sgn}{\mathop{\text{sgn}}\nolimits}

\begin{document}

\begin{frontmatter}

\title{Reduction and efficient solution of MILP models of mixed Hamming packings yielding improved upper bounds}

\author[inst1,inst2]{{P\'eter Naszvadi}}

\affiliation[inst1]{organization={HUN-REN Wigner Research Centre for Physics},
            addressline={Konkoly-Thege Mikl\'os \'ut 29-33.}, 
            city={Budapest},
            postcode={1121}, 
            country={Hungary}}

\affiliation[inst2]{organization={Faculty of Informatics, E\"otv\"os Lor\'and University},
            addressline={P\'azm\'any P\'eter s\'et\'any 1/C}, 
            city={Budapest},
            postcode={1117}, 
            country={Hungary}}

\author[inst1]{{M\'aty\'as Koniorczyk}}

\begin{abstract}
Mixed Hamming packings are considered: the maximal cardinality given a minimum codeword Hamming distance of mixed codes is addressed via mixed integer programming models. Adopting the concept of contact graph from classical continuous sphere packing problems, a reduction technique for the models is introduced, which enables their efficient solution. Several best known upper bounds are improved and some of them are found to be sharp. 
\end{abstract}



\begin{keyword}
mixed codes \sep mixed binary/ternary code \sep Hamming space \sep mixed integer programming \sep densest sphere packing \sep contact graph
\end{keyword}

\end{frontmatter}

\section{Introduction}

Let $\mathbb Z_+$ denote the set of nonnegative integers, $\mathbb Z_m$ be the ring of integers modulo $m$, and $\mathbb Z_m^\alpha$ stand for the ordered $\alpha$-tuples formed from the elements of $\mathbb Z_m$.
A \emph{mixed code} 
is a subset of the Cartesian product ${\mathbb Z}^{\alpha_s}_{k_s} \times \dots \times {\mathbb Z}^{\alpha_2}_{k_2}\times{\mathbb Z}^{\alpha_1}_{k_1}$, with $2\leq k_1<k_2<\dots<k_s$ and $\alpha_j\in \mathbb{Z}, \alpha_j>0$, while $n={\sum_{j=1}^s\alpha_j}$. 
Every element of a code is called a \emph{codeword} or \emph{word}.  To simplify the notation of words the elements of the Cartesian products will be concatenated so that the words are $n$-strings of numbers.
Note that mixed codes are the natural generalizations of \emph{codes} which are the $s=1$ (that is, $n=\alpha_1$) case.
Hence we will use the term \emph{mixed codes} for $n\neq \alpha_1$, i.e. for genuine mixed codes with at least two positions with symbols for a different alphabet, while the term \emph{code} will be used in the general sense for all mixed codes. 
For instance, the \emph{binary-ternary} mixed codes,
with $k_1=2\land k_2=3\land \alpha_1+\alpha_2=n$, are broadly studied in the literature~\cite{WOS:A1991FJ61300010, URL:AEB:Table}. 
They have important applications in football pools~\cite{Virtakallio, WOS:A1995RT88300002} and error correction communication protocols~\cite{WOS:000234944700023}.

Given two words $w$ and $x$, their Hamming-distance is the number of positions in which they differ. 
Formally:
\begin{equation}
    d(w,x)=\sum\limits_{j=1}^{n}\sgn\left|w[j]-x[j]\right|, 
\end{equation}
where $w[j]$ is the $j$-th symbol of the codeword $w$. The set $H$ of all possible codewords, together with the Hamming-distance $d(.,.)$ as a metric is termed as a \emph{Hamming space}. From now on we will use $H$ to denote the Hamming space, and by \emph{distance} we will mean Hamming-distance.

For a code $C$, the \emph{minimum distance of the code} is the minimum of $d(w,x)$ over all distinct $w,x\in C$.
\begin{equation}
    d(C)=\min\limits_{\small\begin{matrix}w\neq x\\ w,x\in C\end{matrix}}d(w,x)
\end{equation}
Amongst the codes having a given minimum distance $d$, there exists at least one which has a maximal cardinality. Such a code is termed as Hamming packing. 
This name comes from a geometric interpretation: considering the codewords as points, with the metric generated by the Hamming distance, a Hamming packing is an extremal packing.
The maximal cardinality a Hamming packing with minimum distance $d$ is denoted by $N_{k_1,k_2,\dots,k_s}(\alpha_1,\alpha_2,\dots,\alpha_s;d)$.

It is an important set of codes for which $k_s=s+1$, implying that the cardinalities $k$ of the subsequent alphabets increase by one (e.g. $k_1=2,\ k_2=3, \ldots$), hence, the $k$-indices of $N$ can be omitted. For instance, in the case of the the aforementioned binary-ternary codes it is practical to introduce the notation $b=\alpha_1$, $t=\alpha_2$, and $N(b,t;d):=N_{2,3}(\alpha_1,\alpha_2;d)$. 

The determination of this maximal cardinality is a key question in coding theory. Usually lower bounds are corollaries of explicit constructions~\cite{WOS:A1977DS09900002}, while upper bounds are derived from estimations, discussed in~\cite{WOS:000242116400012, WOS:000431162700022}. Exact values are known only in a few marginal cases e.g. perfect codes~\cite{Virtakallio, WOS:A1973O900700010, WOS:A1993KH67900025} or for codes in certain very small Hamming-spaces~\cite{URL:AEB:Table, WOS:A1978EN16200010, WOS:A1980KW76600016, WOS:A1991FJ61300009, WOS:A1992HG00600001, WOS:000071193500010}.

In this paper we adopt the notion of the \emph{contact graph} ($CG$) introduced originally for Tammes' $S^2$ classic sphere packing problem~\cite{schutte1951kugel}, for a code $C$ of our code-theoretic setting. The contact graph helps with the classification of ball packings. On this basis, we give a constructive algorithmic proof for the existence of a code having connected contact graph for every $N(\alpha_1,\alpha_2,\dots;d)$. Then we build up an improved mixed integer linear programming (MILP) model that searches for maximal packings having a cherry subgraph (one with and internal node with exactly two leaves) in their contact graphs. Note that the MILP model is equivalent to exhaustive search. State-of-art MILP solvers can solve our model efficiently in many cases, yielding computational bounds that improve some of the currently known best upper bounds to $N(7,1;3)=26$ and $N(4,3;3) = 28$ - former records were~\cite{WOS:000071193500010}: $26\leq N(7,1;3) \leq 30$ and $28\leq N(4,3;3)\leq 30$ (see also ~\cite{URL:AEB:Table} for a regularly updated exhaustive summary of best known bounds). In the last part of this paper, we provide a table for the improved bounds yielded by corollaries.
Our results arise from an integer linear program which, as we shall describe in detail later, differs from the Delsarte~\cite{WOS:000071193500010} linear program; they are for different purposes, yet still provide both lower and upper bounds. 

This paper is organized as follows. In Section~\ref{sec:MILP} we describe the MILP models we study. Section~\ref{sec:reduction} describes our main idea for the size reduction of the models. Section~\ref{sec:compresults} presents our computational results, including the improvement of a number of best known upper bounds. In Section~\ref{sec:conclusion} the results are summarized and conclusions are drawn.

\section{MILP models for Hamming packing}
\label{sec:MILP}

Let $H$ be a given Hamming space and $d\in\mathbb{Z}^+$. Let $x_a$ denote a boolean variable for each possible codeword $a\in H$ so that $x_a=1$ iff $a$ is in the $C$ (maximal) Hamming packing. It is prevalently known that a $C$ can be trivially obtained from the following integer linear program as any of its primal optimal solutions:
\begin{eqnarray}
    \label{eq:milpmodel1}
    \max & \sum\limits_{a\in H}{x_a}& \\
    s.t. &x_a + x_b \leq 1 & 
    \genfrac{}{}{0pt}{1.0}{\forall a,b\in H:}{1\leq d(a,b)\leq (d-1)}
    \nonumber \\
    & x\in{\left\{0,1\right\}}^n& \nonumber.
\end{eqnarray}
The objective ensures the maximality of the packing, whereas the inequality constraints forbid the simultaneous selection of points too close to each other.

The simple model in Eq.~\eqref{eq:milpmodel1} is inefficient in determining $N_{k_1,k_2,\dots,k_s}(\alpha_1,\alpha_2,\dots,\alpha_s;d)$ for the exponential scaling of the number of variables. Without the loss of generality, however, it can be assumed that the all-zero codeword $z$ is always chosen, hence, $z\in C$. Under this assumption the model can be simplified by omitting the variables corresponding to the inner part of the $z$-centered ball with radius $d-1$, leading to
\begin{eqnarray}
    \label{eq:milpmodel2}
    \max & \sum\limits_{a\in H}{x_a}& \\
    s.t. &x_z = 1 \nonumber \\
    ~ & x_a = 0 & \scaleto{1\leq d(a,z) \leq d-1}{7pt} \nonumber \\
    ~ &x_a + x_b \leq 1 & \genfrac{}{}{0pt}{1.0}{\forall a,b\in H: }{1\leq d(a,b)\leq (d-1)} \nonumber \\
    & x\in{\left\{0,1\right\}}^n& \nonumber.
\end{eqnarray}

For practical reasons, the obsoleted binary variables can be removed from the model in order to decrease the memory consumption of the presolvers: the model can be written in the form
\begin{eqnarray}
    \label{eq:milpmodel3}
    \max & 1+\sum\limits_{\genfrac{}{}{0pt}{1.0}{a\in H}{d(a,z) \geq d}}{x_a}& \\
    s.t. &x_z = 1 \nonumber \\
    ~ &x_a + x_b \leq 1 &
    \genfrac{}{}{0pt}{1.0}{\forall a,b\in H\setminus B(z,d-1): }{ 1\leq d(a,b)\leq (d-1)}\nonumber \\
    & x\in{\left\{0,1\right\}}^n& \nonumber.
\end{eqnarray}
Note that developing linear programming models in order to get bounds for the cardinalities of extremal Hamming packings is a fairly old idea, see e.g. Delsarte's results\cite{Delsarte1972Bounds}. Our MILP model, however, has the advantange that it is an equivalent declarative reformulation of the exhaustive search.

\section{Reducing MILPs using contact graphs}
\label{sec:reduction}

In this Section we use contact graphs. We show that every Hamming packing $C$ has an equivalent one, with connected $CG(C)$ contact graph.

\begin{definition} 
The contact graph $CG(C)$  of a Hamming packing $C$ (with minimal distance $d$ which is fixed) is the graph whose vertices correspond to codewords in $C$, and there is an edge between pairs of nodes $a$ and $b$ iff  $d(a,b)=d$.
\end{definition}

\begin{proposition}
\label{prop1}
Every Hamming packing $C$ with minimal distance $d$ can be transformed to another Hamming packing $C'$ with the same number of codewords and minimal distance, whose contact graph $CG(C')$ is connected.
\end{proposition}
\begin{proof}
The desired transfromation can be carried out using the following algorithm ($\bigcup^{*}$ denotes disjoint union).
\begin{algorithm}
    \caption{Transforming $C$ to $C'$ with the same number of codeword and minimal distance, and connected $CG(C')$}\label{alg:proof1}
    \begin{algorithmic}[1]
        \STATE INPUT: $C \subseteq H$ nonempty Hamming packing, $d \in {\mathbb Z}, d>0$ distance
        \STATE START
        \STATE pick a random $\hat w\in C$, them partition $C:=C_1 \bigcup^{*} C_2$, where $C_1=\{\hat w\}$ and $C_2 = C \setminus C_1$
        \WHILE{$C_{2} \neq \emptyset$}
            \WHILE{$\exists w,w': w\in C_1, w'\in C_2: d(w,w')=d$}
                \STATE update $C_1: C_1 := C_1 \bigcup^{*} \{w'\}$
                \STATE update $C_2: C_2 := C_2 \setminus \{w'\}$
            \ENDWHILE \label{lst:line:endwhile2}
            \IF{$(C_{2} \neq \emptyset)$ and $(d(C_1,C_2)\neq d)$}
                \STATE sort increasing the elements of $C_2$ on distance of $C_1$
                \STATE let $d'$ denote: $d':=d(C_1,C_2)$ \label{lst:line:setsetsdistance}
                \STATE select a $w'$ smallest element from $C_2$
                \STATE select arbitrary $w: w\in C_1$, where $d(w,w')=d'$
                \STATE select arbitrary index $j$, where $w[j] \neq w'[j]$
                \STATE denote $a:=w[j], b:=w'[j]$
                \FOR{each $y \in C_2$}
                    \IF{$y[j]=a$}
                        \STATE $y[j]=b$
                    \ELSIF{$y[j]=b$}
                        \STATE $y[j]=a$
                    \ENDIF
                \ENDFOR \label{lst:line:endfor}
            \ENDIF
        \ENDWHILE \label{lst:line:endwhile1}
        \STATE STOP, OUTPUT: $C' := C_1$
    \end{algorithmic}
\end{algorithm}

It is important to note that Algorithm~\ref{alg:proof1} terminates in a finite number of steps. In order to prove this, observe that the following statements are met:
\begin{itemize}
    \item The $C$ code remains a $d$-feasible Hamming packing at every loops' ends in lines \ref{lst:line:endwhile2},~\ref{lst:line:endfor},~\ref{lst:line:endwhile1}.
    \item $CG(C_{1})$ is a connected graph in line \ref{lst:line:endwhile2}.
    \item Once $d'$ reached $d+1$ at line \ref{lst:line:setsetsdistance}, at the end of \emph{while} loop at line \ref{lst:line:endwhile2} $C_2$ must loose one of its elements that will move to $C_1$.
    \item At the end of the \emph{for} loop (line \ref{lst:line:endfor}), $CG(C_2)$ remain isomorph, no Hamming-distance will be altered between elements of $C_2$.
    \item At the end of the for loop the distance betveen $C_1$ and $C_2$ decreases exactly by one. As the distance was larger than d at the beginning of the loop,
$C_1 \bigcup^\ast C_2$ still does not contain any pairs with distance less than $d$.
    \item The cardinality of $C_2$ is strictly decreasing or $d(C_1,C_2)$ is strictly decreasing in line in line \ref{lst:line:endwhile2},~\ref{lst:line:endfor}~\ref{lst:line:endwhile1}.
\end{itemize}
    As the cardinality of $C_2$ is finite and $d(C_1,C_2) \leq n$, the outer loop always terminates. 
\end{proof}

As a consequence of Propositon~\ref{prop1}, the discussion can be restricted to Hamming packings with connected contact graphs. Let us consider the case when $|C|\geq 2$. Then we can fix two nodes which are adjacent. In the special case when $H = \mathbb{Z}^n_k$ it is enough to force the selection of two balls when determining $N_k(n;d)$. However, in the case of general mixed codes, we can also choose the all-zero codeword as one of the nodes, but we must partition the nonzero elements amongst the same-alphabet subwords of the code. Hence, we must branch when selecting the second node of the pair. Upon this branching it would be a mistake to assign all the nonzero letters to the smallest-alphabet subwords: table \ref{tab:counterexamples-d3} contains some of the counterexample optimal mixed codes with length $5$ and $d=3$, whereas \ref{tab:counterexamples-d4} tabulates counterexamples for $d=4$. These examples illustrate the necessity of traversing other branches, too.

The key idea of our reduction technique is the following: \emph{in each branch, having chosen the initial pair of codewords, the words that are at a distance less than $(d-1)$ from the chosen pair, can be eliminated from the model by fixing the respective variables.} This leads to a drastical reduction in the size of the respective MILP models.

\hspace{0pt}

\begin{table}[ht!]
\centering
\begin{tabular} {|c||c|c|}\hline
 & \makecell{Cardinality of maximal \\Hamming-packing} & \makecell{Cardinality of maximal \\Hamming-packing \\ when $00000$ and $00111$ codewords \\ must be selected} \\\hline\hline
$N(4,1;3)$&$6$&$4$\\\hline
$N(2,3;3)$&$9$&$8$\\\hline
$N_{2,4}(4,1;3)$&$8$&$5$\\\hline
$N(2,2,1;3)$&$11$&$9$\\\hline
$N_{2,5}(4,1;3)$&$8$&$5$\\\hline
$N_{2,3,5}(2,2,1;3)$&$12$&$11$\\\hline
$N_{2,6}(4,1;3)$&$8$&$5$\\\hline
$N_{2,7}(4,1;3)$&$8$&$5$\\\hline
\end{tabular}
\caption{\label{tab:counterexamples-d3}Counterexamples for $d=3$.}
\end{table} ~ \begin{table}[ht!]
\centering
\begin{tabular} {|c||c|c|}\hline
 & \makecell{Cardinality of maximal \\Hamming-packing} & \makecell{Cardinality of maximal \\Hamming-packing \\ when $00000$ and $01111$ codewords \\ must be selected} \\\hline\hline
$N(3,2;4)$&$3$&$2$\\\hline
$N(3,1,1;4)$&$3$&$2$\\\hline
$N(2,2,1;4)$&$4$&$3$\\\hline
$N_{2,4}(3,2;4)$&$4$&$2$\\\hline
$N_{2,3,5}(3,1,1;4)$&$3$&$2$\\\hline
$N_{2,3,5}(2,2,1;4)$&$4$&$3$\\\hline
$N_{2,3,5}(1,3,1;4)$&$5$&$4$\\\hline
$N_{2,4,5}(3,1,1;4)$&$4$&$2$\\\hline
$N_{2,5}(3,2;4)$&$4$&$2$\\\hline
$N_{2,3,6}(3,1,1;4)$&$3$&$2$\\\hline
$N_{2,3,6}(2,2,1;4)$&$4$&$3$\\\hline
$N_{2,3,6}(1,3,1;4)$&$6$&$4$\\\hline
$N_{2,4,6}(3,1,1;4)$&$4$&$2$\\\hline
$N_{2,5,6}(3,1,1;4)$&$4$&$2$\\\hline
$N_{2,6}(3,2;4)$&$4$&$2$\\\hline
$N_{2,3,7}(3,1,1;4)$&$3$&$2$\\\hline
$N_{2,3,7}(2,2,1;4)$&$4$&$3$\\\hline
$N_{2,3,7}(1,3,1;4)$&$6$&$4$\\\hline
$N_{3,7}(4,1;4)$&$7$&$6$\\\hline
$N_{2,4,7}(3,1,1;4)$&$4$&$2$\\\hline
$N_{3,8}(4,1;4)$&$8$&$6$\\\hline
$N_{3,9}(4,1;4)$&$9$&$6$\\\hline
\end{tabular}
\caption{\label{tab:counterexamples-d4}Counterexamples for $d=4$.}
\end{table}

\hspace{0pt}

\section{Computational results}
\label{sec:compresults}

For our experiments we have used the following hardware configuration: 256GB RAM, two AMD EPYC 7302 16-Core Processors, total 32 cores. Version 22.1.0.0 of IBM CPLEX~\cite{SW:CPLEX2210} was used.

\subsection{Finding N(7,1;3)=26}

For mixed binary-ternary codes, the Hamming-distance can be written as a sum of two functions,
\begin{eqnarray}
    \label{eq:hammingdecomp}
    d(w,w') &=&d^b(w,w')+d^t(w,w'),
\end{eqnarray}
where $d^b(.,.)$ is the number of differences at binary, while $d^t(.,.)$ is the differences at ternary coordinates. The lower bound had been already known via an explicit construction according to~\cite{WOS:000071193500010}, and the upper bound was $30$.

\begin{lemma}
\label{lemma1}
For $N(7,1;3)$, every maximal Hamming packing must contain two codewords $w, w'$, where the following property is met: $d(w,w')=3$ and $d^b(w,w')=3$.
\end{lemma}

\begin{proof}
Indirectly suppose that there is a counterexample Hamming packing with at least $26$ codewords. This packing also must be a feasible solution of the following ILP model:
\begin{eqnarray}
    \label{eq:milpmodel713}
    \max & \sum\limits_{{a\in H}}{x_a}& \\
    s.t. &x_a + x_w \leq 1 &
    {1\leq d(a,w)\leq 2}\nonumber \\
    ~ &x_a + x_w \leq 1 & {d(a,w)=3 \land d^b(a,w)=3}\nonumber \\
    & x\in{\left\{0,1\right\}}^n& \nonumber
\end{eqnarray}
Both \emph{GLPK}~\cite{SW:GLPK50} and \emph{CPLEX}~\cite{SW:CPLEX2210} solved the model \eqref{eq:milpmodel713} fast, that is, in less than 5 minutes, getting an integer optimal solution with an objective value $24$ strictly below $26$ - the best known feasible solution of the original problem. 
\end{proof}

\paragraph{Computational model and evaluation.}

As a consequence of Lemma~\ref{lemma1}, it is sufficient to consider one of the two branches.
Hence, we solve the MILP model described at \eqref{eq:milpmodel3} with the following constraints fixing the initial pair:  
\begin{eqnarray}
    \label{eq:milpmodel713b}
x_{0,0000000}=1, \nonumber \\ x_{0,0000111}=1.
\end{eqnarray}
The so reduced model contains $300$ binary variables. After five days of running, \emph{CPLEX} exited providing an integer optimal solution.

An example of a solution is the following list of total 26 codewords:
$00000000$, $00000111$, $00011001$, $00011110$, $00110011$, $01001101$, $01101010$, $01110100$, $10001100$,
$10010010$, $10100110$, $10101001$, $10110101$, $11001011$, $11010111$, $11011000$, $11100000$, $11111110$,
$20001010$, $20010100$, $20101111$, $20111000$, $21000001$, $21000110$, $21110010$, $21111101$. \hfill $\square$

\subsection{Finding N(4,3;3)=28}

Again, the marginal metric decomposition will be used as defined previously in \eqref{eq:hammingdecomp}. The lower bound had been already known via an explicit construction~\cite{WOS:000071193500010}, and the upper bound was $30$.

\paragraph{A specially constrained model}
For $N(4,3;3)$, every maximal Hamming packing must contain two codewords $w, w'$, where the following property is met: $d(w,w')=3$ and $d^t(w,w')=3$. To prove this, the pigeonhole principle should be applied: arbitrary $17$ elements from every packing must have at least two codes containing exactly the same binary digits, which implies that all further $3$ (ternary) digits must differ in the selected codeword pair. Using this observation, the MILP model in \eqref{eq:milpmodel3} had been solved with two additional constraints
\begin{eqnarray}
    \label{eq:milpmodel433b}
x_{000,0000}=1, \nonumber \\ x_{111,0000}=1. 
\end{eqnarray}
The solving process successfully finished using \emph{CPLEX}, which exited with integer optimal status after $6$ hours. It returned the objective function value of  $28$. \hfill $\square$

An example solution with $28$ codewords is the following: $0000000$, $0010101$, $0011010$, $0021111$, $0100110$, $0101011$, $0111100$, $0201101$, $0210011$, $0221000$, $1001100$, $1010110$, $1020001$, $1100101$, $1110000$, $1111111$, $1121010$, $1200010$, $1211001$, $1220111$, $2000111$, $2020010$, $2101000$, $2121101$, $2200001$, $2201110$, $2210100$, $2221011$.

\subsection{Further improvements for $N(b,t;d)$}

The improved upper bounds we have found for mixed binary-ternary Hamming-packings with distance $d=3$ are tabulated in table \ref{tab:summary-results-d3} whereas for distances $d=4$ in table \ref{tab:summary-results-d4}. In the tables, each nonempty cell contains the former upper bound in parentheses, and inequality is printed iff our improved upper bound is still greater than the known biggest lower bound, that is, the respective MIP gap is still not closed.

\hspace{0pt}

\begin{table}[ht!]
\centering
\small
\begin{tabular} {|c||c|c|c|c|c|}\hline
$b$ \textbackslash $t$ &1&2&3&4&5 \\\hline\hline
2&&&&&$\leq 63~(65)^{\text{[pr. 4.3/iv]}}$\\\hline
3&&&&$= 42~(44)^{\text{[pr. 4.3/iv]}}$&\\\hline
4&&&= 28 (30)&$\leq 84~(88)^{\footnotesize[pr. 4.3/ii]}$&\\\hline
5&&&$\leq 56~(60)^{\text{[pr. 4.3/ii]}}$&&\\\hline
6&&$\leq 39~(44)^{\text{[pr. 4.3/iv]}}$&$\leq 112~(118)^{\text{[pr. 4.3/ii]}}$&&\\\hline
7&= 26 (30)&$\leq 78~(83)^{\text{[pr. 4.3/ii]}}$&$\leq 224~(225)^{\text{[pr. 4.3/ii]}}$&&\\\hline
8&$\leq 52~(59)^{\text{[pr. 4.3/ii]}}$&&&&\\\hline
9&$\leq 104~(108)^{\text{[pr. 4.3/ii]}}$&&&&\\\hline
10&$\leq 208~(212)^{\text{[pr. 4.3/ii]}}$&&&&\\\hline
\end{tabular}
\caption{\label{tab:summary-results-d3}Improved upper bounds of Hamming packings for $d=3$. The number in the parentheses reflect the former best known upper bounds, followed by the proposition's number in~\cite{WOS:000071193500010} they are known from.}
\end{table}

\begin{table}[ht!]
\centering
\small
\begin{tabular} {|c||c|c|c|c|}\hline
$b$ \textbackslash $t$ &2&3&4&5\\\hline\hline
3&&&&$\leq 42~(43)^{\text{[prop. 4.3/iv]}}$\\\hline
4&&&$ = 28~(30)^{\text{[prop. 4.3/vi]}}$&\\\hline
5&&&$\leq 56~(59)^{\text{[prop. 4.3/ii]}}$&\\\hline
6&&$\leq 39~(40)^{\text{[prop. 4.3/iv]}}$&$\leq 112~(114)^{\text{[prop. 4.3/ii]}}$&\\\hline
7&$= 26~(30)^{\text{[prop. 4.3/vi]}}$&$\leq 78~(80)^{\text{[prop. 4.3/ii]}}$&&\\\hline
8&$\leq 52~(59)^{\text{[prop. 4.3/ii]}}$&&&\\\hline
9&$\leq 104~(108)^{\text{[prop. 4.3/ii]}}$&&&\\\hline
10&$\leq 208~(212)^{\text{[prop. 4.3/ii]}}$&&&\\\hline
\end{tabular}
\caption{\label{tab:summary-results-d4}Improved upper bounds of Hamming packings for $d=4$. The number in the parentheses reflect the former best known upper bounds, followed by the proposition's number in~\cite{WOS:000071193500010} they are known from.}
\end{table}

\section{Summary and outlook}
\label{sec:conclusion}

We have introduced a reduction technique for mixed codes. Our technique is based on our idea to adopt the notion of contact graphs, motivated by continuous sphere packing problems. Using the technique, we have efficiently improved various best known values of maximal cardinalities of Hamming packings with a given minimum distance of Hamming packings, using mixed integer linear programming.
Our approach can work for bigger problem instances, and is not restricted to binary-ternary codes. In spite of the limited number of variables in the models they are challenging for classical solvers. This suggests that their further study may yield benchmark problems for quantum computers that bear practical relevance.

\section*{Acknowledgements}

This research was supported by the National Research, Development, and Innovation Office of Hungary under project numbers K133882 and K124351, the Ministry of Innovation and Technology and the National Research, Development and Innovation Office within the Quantum Information National Laboratory of Hungary.
The hardware resources were provided by the Wigner Scientific Computing Laboratory (WSCLAB). We thank Mikl\'os Pint\'er (Corvinus University, Budapest) for running the CPLEX calculations, and S\'andor Szab\'o (University of P\'ecs) and Andr\'as Bodor (HUN-REN Wigner RCP) for useful discussions.


\end{document}